\documentclass[conference]{IEEEtran}
\usepackage{fancyhdr}
\usepackage{graphicx}
\usepackage[utf8]{inputenc}
\usepackage{amsmath}
\usepackage[caption=false,font=footnotesize]{subfig}
\usepackage{amsfonts}
\usepackage{amssymb}
\usepackage{amsthm}
\usepackage{hyperref}
\usepackage{breqn}

\usepackage{cite}
\usepackage{epsf,bm}
\usepackage{epstopdf}
\newtheorem{theorem}{Theorem}

\newtheorem{lemma}{Lemma}
\newtheorem{corollary}{Corollary}
\usepackage{changepage}

\IEEEoverridecommandlockouts
\begin{document}

\title{Cache-Aided Heterogeneous Networks: Coverage and Delay Analysis}

\author{\large Mohamed A. Abd-Elmagid$^*$$^\dagger$, Ozgur Ercetin$^\dagger$, and Tamer ElBatt$^*$$^\ddagger$ \\ [.05in]
\small  \begin{tabular}{c} $^*$Wireless Intelligent
Networks Center (WINC), Nile University, Giza, Egypt.\\$^\dagger$ Faculty of Engineering and Natural Sciences, Sabanci University, Istanbul, Turkey. \\
$^\ddagger$Dept. of EECE, Faculty of Engineering, Cairo University, Giza, Egypt.
\end{tabular}
\thanks{This work was supported in part by the European Unions Horizon 2020 research and innovation programme under the Marie Skodowska-Curie grant agreement No 690893.}
}
\maketitle

\begin{abstract}
 This paper characterizes the performance of a generic $K$-tier cache-aided heterogeneous network (CHN), in which the base stations (BSs) across tiers differ in terms of their spatial densities, transmission powers, pathloss exponents, activity probabilities conditioned on the serving link and placement caching strategies. We consider that each user connects to the BS which maximizes its average received power and at the same time caches its file of interest. Modeling the locations of the BSs across different tiers as independent homogeneous Poisson Point processes (HPPPs), we derive closed-form expressions for the coverage probability and local delay experienced by a typical user in receiving each requested file. We show that our results for coverage probability and delay are consistent with those previously obtained in the literature for a single tier system.
\end{abstract}

\section{Introduction}
Cache-aided small cell networks (CSNs) have recently attracted considerable attention in the literature \cite{golrezaei2013femtocaching,blaszczyszyn2015optimal, yang2016analysis,bacstu2015cache,chen2016probabilistic}. \cite{golrezaei2013femtocaching} proposed a new caching architecture, namely, Femtocaching, in which the small cell base stations (SBSs) are utilized as distributed caching devices with either small or non-existing backhaul capacity but with considerable storage space. Attributed to the nature of randomness in both the mobile user locations and the available contents inside caches, stochastic geometry is considered to be a very relevant tool for the analysis of CSNs \cite{blaszczyszyn2015optimal, yang2016analysis, bacstu2015cache, chen2016probabilistic}. In \cite{blaszczyszyn2015optimal}, a probabilistic placement caching strategy is adopted at each SBS, in which each of the popular files is cached with a certain probability. The prime goal was to characterize the optimal caching probabilities for different files to maximize the hit probability subject to the finite size caches constraints. Modeling the distributions of mobile users and SBSs as HPPPs facilitates the ability to obtain closed-form expressions of various performance metrics for different setups of CSNs \cite{yang2016analysis, bacstu2015cache, chen2016probabilistic}, e.g., coverage probability and average achievable rate. \cite{yang2016analysis} considered device-to-device (D2D) caching along with small cell caching through the assumption that a subset of the mobile terminals are equipped with finite cache sizes. The SBSs were assumed to cache the same copy of specific popular contents in \cite{yang2016analysis, bacstu2015cache}. On the contrary, \cite{chen2016probabilistic} considered that the content library is  partitioned into non-overlapping subsets of files, and each SBS can only cache one of them with a certain probability, which provides a more flexibility compared to caching the same copy of contents. Differently from prior works, in this paper we characterize the performance of a generic $K$-tier \textit{cache-aided heterogeneous network} (CHN), in which each tier may adopt a different probabilistic placement caching strategy.

Our main contribution in this paper is to characterize the performance of $K$-tier CHNs. We study a generic model for a $K$-tier CHN, in which the BSs across tiers are distinguished by their spatial densities, transmission powers, pathloss exponents and activity probabilities conditioned on the serving link. Moreover, we consider the generic scenario in which the BSs across different tiers may employ different placement caching strategies. Modeling the locations of BSs in each tier as an independent HPPP, we derive closed-form expressions for the coverage probability and local delay experienced by a typical user while trying to obtain each cached file in the network. We further obtain a simplified expressions for the special case of having identical pathloss exponents across tiers, through which we show the consistence of our obtained results with prior works for the single tier heterogeneous network scenario.
\section{System Model}
We study a generic model for a CHN consisting of $K$ independent network tiers. The BSs across tiers are distinguished by their spatial densities, transmission powers, pathloss exponents, activity probabilities conditioned on the serving link and placement caching strategies. Particularly, the locations of BSs in the $j$-th tier are modeled as a HPPP $\phi_{j}$ with spatial density $\lambda_{j}$. We consider a content library composed of $M$ different files and assume that each tier's BSs may adopt a different probabilistic placement caching strategy. Let $\mathbf{p}_{j} = \{p_{1j},p_{2j},\ldots,p_{Mj}\}$ denote the probabilistic placement caching strategy employed by the $j$-th tier's BSs, where $p_{mj}$ denotes the probability that a j-th tier's BS caches file $m$ such that $\sum_{m=1}^{M}{p_{mj}} = S_{j} \;\text{files}, \forall j$, where $S_{j}$ is the cache size associated with the $j$-th tier's BSs. Conditioned on the serving link, each interferer from the $j$-th tier is assumed to be independently active with probability $a_{j}$. This is due to the fact that the BSs may not be always active to save energy or to mitigate interference, etc.

 According to Slivnyak's theorem \cite{haenggi2012stochastic}, an existing point in the process does not change the statistical distribution of other points of the PPP. Therefore, without loss of generality, we focus on the downlink analysis at a typical user located at the origin. Given that the typical user is associated with BS $i$ in the $j$-th tier located at $Y_{ji} \in \phi_{j}$, the received power at the typical user is given by

\vspace{-0.3cm}
\small
\begin{align}
P_{j}^{\text{recv}} = P_{j} h_{ji} \parallel Y_{ji} \parallel^{-\alpha_{j}},
\end{align}
\normalsize
where $P_{j}$ denotes the transmission power of the $j$-th tier, $\alpha_{j}$ denotes the pathloss exponent of the $j$-th tier and $h_{ji}$ is the channel power gain from BS $i$ of the $j$-th tier and the typical user. We assume independent Rayleigh fading channel coefficients with unit average power, i.e., $h_{ji}$ is an exponential random variable with unit mean.

We assume an open access network, i.e., the typical user is allowed to connect to any tier without any restrictions. We consider a \textit{cell association strategy}\footnote{Note that $P_{j} R_{nj}^{-\alpha_{j}}$ is a long-term average of $P_{j}^{\text{recv}}$ and fading is averaged out, and hence it does not include $h_{ji}$.}, where the typical user connects to the BS which offers the strongest \textit{average} received power and caches its requested file. For instance, the index of the tier of the BS serving the user request file $n$ is given by

\vspace{-0.3cm}
\small
\begin{align}\label{cas}
k = \text{arg}\; \max_{\substack{j \in \mathcal{K}}} \;P_{j} R_{nj}^{-\alpha_{j}},
\end{align}
\normalsize
where $\mathcal{K} = \{1,2,\cdots,K\}$ and $R_{nj}$ denotes the distance between the user and the closest BS in the $j$-th tier that caches file $n$. 

Throughout this paper, we emphasize our analysis on a scenario where the typical user tries to obtain file $n$ cached in the network. According to the thinning theorem of the PPP \cite{haenggi2012stochastic}, the distribution of the $k$-th tier's BSs that caches file $n$ can be viewed as a thinned HPPP $\phi_{k}^{n}$ with density $p_{nk} \lambda_{k}$. For instance, for a typical user associated with a BS $\in \phi_{k}^{n}$ located at a random distance $x_{nk}$ from the user, the signal to interference ratio (SIR) at the typical user is given by

\vspace{-0.3cm}
\small
\begin{align}
\text{SIR}_{k}(x_{nk}) = \dfrac{P_{k} h_{k0} x_{nk}^{-\alpha_{k}}}{\sum_{j=1}^{K}\sum_{i \in \phi_{j} \setminus B_{k0}}{t_{ji} P_{j} h_{ji} \parallel Y_{ji} \parallel^{-\alpha_{j}}}},
\end{align}
\normalsize
where $B_{k0}$ is the index of the serving BS and $t_{ji}$ is an indicator random variable which represents the activity of BS $i$ in the $j$-th tier. Thus, $t_{ji}$ takes value $1$ with probability $a_{j}$ and $0$ otherwise. Note that we assume the network is interference limited, and hence the thermal noise power is ignored compared to interference power.
%
\section{Coverage Probability Analysis}
In this section we derive the coverage probability for the considered $K$-tier CHN. The coverage probability is defined as the probability that a typical user is able to achieve some threshold SIR, denoted $\tau$, when the user tries to obtain its file of interest from its associated BS. Recall that we focus on a scenario where the typical user tries to obtain file $n$. Since the typical user is associated with at most one tier, from the total probability law, the coverage probability of obtaining file $n$ can be expressed as
\begin{align}\label{total.cov}
C_{n} = \sum_{k=1}^{K}{A_{k} C_{nk}},
\end{align}
where $A_{k}$ is the probability that the typical user is associated with the $k$-th tier and $C_{nk}$ is the probability of coverage when the user is associated with the $k$-th tier to obtain file $n$. Thus, $C_{nk}$ is given by

\small
\begin{align}\label{cov.nk}
\nonumber C_{nk} &= \mathbb{E}_{x_{nk}}\left[{\mathbb{P}\left(\text{SIR}_{k}\left(x_{nk}\right) > \tau \left.\right| x_{nk}\right)}\right] \\
\nonumber &= \int_{0}^{\infty} \mathbb{P}\left(h_{k0} > \dfrac{x_{nk}^{\alpha_{k}}\tau I}{P_{k}}\left.\right| x_{nk}\right) f_{X_{nk}}(x_{nk}) \;\text{d}x_{nk} \\
\nonumber &\overset{(a)}{=} \int_{0}^{\infty} \mathbb{E}_{I}\left[\text{exp}\left(-\dfrac{x_{nk}^{\alpha_{k}}\tau I}{P_{k}}\right)\right] f_{X_{nk}}(x_{nk}) \;\text{d}x_{nk} \\
 &\overset{(b)}{=} \int_{0}^{\infty} \prod_{j=1}^{K}{\mathcal{L}_{I_{j}}\left(\dfrac{x_{nk}^{\alpha_{k}}\tau}{P_{k}}\right)} f_{X_{nk}}(x_{nk}) \;\text{d}x_{nk},
\end{align}
\normalsize
where $I_{j}$ denotes the interference power of the $j$-th tier such that $I = \sum_{j=1}^{K}{I_{j}}$, $\mathcal{L}_{I_{j}}$ denotes the Laplace transform of $I_{j}$, $f_{X_{nk}}(x_{nk})$ is the probability density function (PDF) of the distance between the typical user and the serving BS in the $k$-th tier. Note that (a) follows from $h_{k0} \sim \text{exp}(1)$ and (b) follows from the definition of Laplace transform along with that fact that $I_{j}$, $j=1,2,\cdots,K,$ are independent random variables. Before proceeding to derive the Laplace transform of each tier, we provide the PDF of the distance between the typical user and the serving BS in the $k$-th tier, $f_{X_{nk}}(x_{nk})$, in the following Lemma.

\begin{lemma}\label{lemma1}
The PDF of the distance between the typical user and its serving BS $\in \phi_{k}^{n}$, that caches file $n$, is given by
\small
\begin{align}
\nonumber f_{X_{nk}}(x_{nk}) = \dfrac{2 \pi p_{nk} \lambda_{k} x_{nk}}{A_{k}} \text{exp} \left[- \pi \sum_{j=1}^{K}{p_{nj} \lambda_{j} \bar{P}_{j}^{2/\alpha_{j}} x_{nk}^{2/ \bar{\alpha}_{j}}}\right],
\end{align}
\normalsize
where $\bar{P}_{j} = \dfrac{P_{j}}{P_{k}}$ and $\bar{\alpha}_{j} = \dfrac{\alpha_{j}}{\alpha_{k}}$, i.e., $\bar{P}_{k} = \bar{\alpha}_{k} = 1$.
\end{lemma}
\begin{proof}
Note that the event $\{X_{nk} > x_{nk}\}$ is equivalent to the event $\{R_{nk} > x_{nk}\}$ conditioned on the association of the typical user with a $k$-tier BS, thus we have
\begin{align}
\nonumber \mathbb{P}\left(X_{nk} > x_{nk}\right) &= \mathbb{P}\left(R_{nk} > x_{nk}  \left.\right| \text{Association index = k}\right).
\end{align}
The result follows from Lemma $3$ in \cite{jo2012heterogeneous} by observing that the distribution of $R_{nk}$ can be obtained from its null probability as $f_{R_{nk}}(r_{nk}) = 2\pi p_{nk} \lambda_{k} r_{nk} \text{exp}\left(-\pi p_{nk} \lambda_{k} r_{nk}^{2}\right)$ \cite{haenggi2012stochastic}.
\end{proof}
Now, our objective is to derive the Laplace transform of the interference power $\mathcal{L}_{I_{j}}\left(\dfrac{x_{nk}^{\alpha_{k}}\tau}{P_{k}}\right)$ in order to evaluate the coverage probability (\ref{cov.nk}). We first derive the Laplace transform of the serving tier's interference power, i.e., $\mathcal{L}_{I_{k}}$, and then we derive the Laplace transform of the other tiers' interference powers. Since the serving BS in the $k$-th tier located at distance $x_{nk}$ from the typical user, none of the $k$-th tier's BSs which are located inside the circle of radius $x_{nk}$ cache file $n$. Therefore, the $k$-th tier's BSs located inside and outside the circle of radius $x_{nk}$ constitute two HPPPs with densities $q_{nk}\lambda_{k}$ and $\lambda_{k}$, respectively, and with activity probability $a_{k}$, where $q_{nk} = 1 - p_{nk}$. $\mathcal{L}_{I_{k}}$ can be determined by evaluating the Laplace transform of the $k$-th tier's interference power from outside and inside the circle of radius $x_{nk}$ separately.
\begin{lemma}\label{lemma2}
The Laplace transform of the $k$-th tier's interference power from outside the circle of radius $x_{nk}$ is
\small
\begin{align}
\mathcal{L}_{I_{k}}^{\text{out}}\left(\dfrac{x_{nk}^{\alpha_{k}}\tau}{P_{k}}\right) = \text{exp}\left[-\pi \lambda_{k} a_{k} \rho_{1}(k) x_{nk}^{2}\right],
\end{align}
\normalsize
where $\rho_{1}(m) = \tau^{2/\alpha_{m}} \int_{\tau^{-2/\alpha_{m}}}^{\infty} \dfrac{1}{1 + u^{\alpha_{m}/2}}\; \text{d}u$. 
\end{lemma}
\begin{proof}
Following the same approach used in the proof of Theorem 1 in \cite{andrews2011tractable}, the result can be obtained as follows
\small
\begin{align}
\nonumber \mathcal{L}_{I_{k}}^{\text{out}} &= \mathbb{E}_{\phi_{k},t_{ki},h_{ki}}\left[\prod_{i \in \phi_{k}\setminus B(0,x_{nk})}{\text{exp}\left(-x_{nk}^{\alpha_{k}} \tau t_{ki} h_{ki} \parallel Y_{ki} \parallel^{-\alpha_{k}}\right)}\right] \\
\nonumber &\overset{(a)}{=}\mathbb{E}_{\phi_{k},t_{ki}}\left[\prod_{i \in \phi_{k}\setminus B(0,x_{nk})}{\dfrac{1}{1 + x_{nk}^{\alpha_{k}} \tau t_{ki} \parallel Y_{ki} \parallel^{-\alpha_{k}}}}\right] \\
\label{Iout.mid} &\overset{(b)}{=}\mathbb{E}_{\phi_{k}}\left[\prod_{i \in \phi_{k}\setminus B(0,x_{nk})}{\dfrac{a_{k}}{1 + x_{nk}^{\alpha_{k}} \tau \parallel Y_{ki} \parallel^{-\alpha_{k}}} + (1 - a_{k})}\right]\\
\nonumber &\overset{(c)}{=}\text{exp}\left[- 2\pi \lambda_{k} a_{k} \int_{x_{nk}}^{\infty} \dfrac{\tau}{\tau + (\dfrac{y}{x_{nk}})^{\alpha_{k}}}\; y\text{d}y\right],
\end{align}
\normalsize
where $B(0,x_{nk})$ denotes the set of $k$-th tier's BSs located inside a circle centered at the origin and with radius $x_{nk}$, (a) follows from the independence of the random variables $h_{ki}$ along with the fact that $h_{ki} \sim \text{exp}(1)$, (b) follows from the fact that the random variables $t_{ki}$ are independent and (c) follows from the probability generating functional (PGFL) of the PPP along with simple algebraic manipulations. By using the change of variables $u = \left(y/x_{nk}\right)^{2} \tau^{-2/\alpha_{k}}$, the result can be directly obtained.
\end{proof}
Following the same analysis of Lemma \ref{lemma2}, the Laplace transform of the $k$-th tier's interference power from inside the circle of radius $x_{nk}$ is given by the following Lemma. The proof of Lemma \ref{lemma3} is omitted for brevity.
\begin{lemma}\label{lemma3}
The Laplace transform of the $k$-th tier's interference power from inside the circle of radius $x_{nk}$ is
\begin{align}
\mathcal{L}_{I_{k}}^{\text{ins}}\left(\dfrac{x_{nk}^{\alpha_{k}}\tau}{P_{k}}\right) = \text{exp}\left[-\pi \lambda_{k} a_{k} q_{nk} \rho_{2}(k) x_{nk}^{2}\right],
\end{align}
where $\rho_{2}(m) = \tau^{2/\alpha_{m}} \int_{0}^{\tau^{-2/\alpha_{m}}} \dfrac{1}{1 + u^{\alpha_{m}/2}}\; \text{d}u$. 
\end{lemma}
Exploiting the fact that the $k$-th tier's interference power from outside and inside the circle of radius $x_{nk}$ are independent random variables, the Laplace transform of the serving tier's interference power, $\mathcal{L}_{I_{k}}$, can be expressed as

\small
\begin{align}\label{lap.serv}
\nonumber \mathcal{L}_{I_{k}}\left(\dfrac{x_{nk}^{\alpha_{k}}\tau}{P_{k}}\right) &= \mathcal{L}_{I_{k}}^{\text{out}}\left(\dfrac{x_{nk}^{\alpha_{k}}\tau}{P_{k}}\right) \times \mathcal{L}_{I_{k}}^{\text{ins}}\left(\dfrac{x_{nk}^{\alpha_{k}}\tau}{P_{k}}\right) \\
 &= \text{exp}\left[-\pi \lambda_{k} a_{k} x_{nk}^{2}\left(\rho_{1}(k) + q_{nk} \rho_{2}(k)\right)\right].
\end{align}
\normalsize
The typical user is associated with the $k$-th tier to obtain file $n$ if $R_{nj} > \bar{P}_{j}^{1/\alpha_{j}} x_{nk}^{1/\bar{\alpha}_{j}},\; \forall j \neq k$. Let $x_{j} = \bar{P}_{j}^{1/\alpha_{j}} x_{nk}^{1/\bar{\alpha}_{j}}$. Given that the typical user with a BS at $k$-th tier, the distance from the typical user to the closest BS that caches file $n$ in the $j$-th tier, $\forall j \neq k$, is at least $x_{j}$. Hence, the $j$-th tier's BSs located inside and outside the circle of radius $x_{j}$ constitute two HPPPs with densities $q_{nj} \lambda_{j}$ and $\lambda_{j}$, respectively, and with activity probability $a_{j}$, where $q_{nj} = 1 - p_{nj}$.
\begin{lemma}\label{lemma4}
Given that the typical user associates with the $k$-th tier, the Laplace transform of the $j$-th tier's interference power from outside the circle of radius $x_{j}$ is
\small
\begin{align}
\nonumber \mathcal{L}_{I_{j}}^{\text{out}}\left(\dfrac{x_{nk}^{\alpha_{k}}\tau}{P_{k}}\right) = \text{exp}\left[-\pi \lambda_{j} a_{j} \bar{P}_{j}^{2/\alpha_{j}} \rho_{1}(j) x_{nk}^{2/\bar{\alpha}_{j}} \right],\; \forall j \neq k.
\end{align}
\end{lemma}
\normalsize
\begin{proof}
The result can be obtained in a similar way as in Lemma \ref{lemma2} by replacing the radius of the circle $x_{nk}$ with $x_{j}$ and considering the transmission power of the $j$-th tier $P_{j}$.
\end{proof}
Following the same analysis of Lemma \ref{lemma4}, the Laplace transform of the $j$-th tier's interference power from inside the circle of radius $x_{j}$ is given by the following Lemma.
\begin{lemma}\label{lemma5}
Conditioned on the association of the typical user with the $k$-th tier, the Laplace transform of the $j$-th tier's interference power from inside the circle of radius $x_{j}$ is
\begin{align}
\nonumber \mathcal{L}_{I_{j}}^{\text{ins}}\left(\dfrac{x_{nk}^{\alpha_{k}}\tau}{P_{k}}\right) = \text{exp}\left[-\pi \lambda_{j} q_{nj} a_{j} \bar{P}_{j}^{2/\alpha_{j}} \rho_{2}(j) x_{nk}^{2/\bar{\alpha}_{j}} \right],\; \forall j \neq k.
\end{align}
\end{lemma}
Using Lemma \ref{lemma4} and \ref{lemma5} along with exploiting the fact that the $j$-th tier's interference power from outside and inside the circle of radius $x_{j}$ are independent random variables, the Laplace transform of the $j$-th tier's interference power can be expressed as
\begin{align}\label{lap.other}
\nonumber \mathcal{L}_{I_{j}}\left(\dfrac{x_{nk}^{\alpha_{k}}\tau}{P_{k}}\right) &= \mathcal{L}_{I_{j}}^{\text{out}}\left(\dfrac{x_{nk}^{\alpha_{k}}\tau}{P_{k}}\right) \times \mathcal{L}_{I_{j}}^{\text{ins}}\left(\dfrac{x_{nk}^{\alpha_{k}}\tau}{P_{k}}\right) \\
\nonumber &= \text{exp}\left[-\pi \lambda_{j} a_{j} \bar{P}_{j}^{2/\alpha_{j}}  x_{nk}^{2/\bar{\alpha}_{j}} \left(\rho_{1}(j) + q_{nj} \rho_{2}(j)\right)\right].
\end{align}
\begin{theorem}\label{theorem1}
The coverage probability of obtaining file $n$ when the typical user is associated with the $k$-th tier is
\small
\begin{dmath*}\label{cnk.final}
C_{nk}=\int_{0}^{\infty} \dfrac{2 \pi p_{nk} \lambda_{k} x_{nk}}{A_{k}} \text{exp}\left[-\pi \sum_{j=1}^{K}{\lambda_{j} \bar{P}_{j}^{2/\alpha_{j}} x_{nk}^{2/ \bar{\alpha}_{j}}\left( p_{nj} + a_{j} \left( \rho_{1}(j) + q_{nj} \rho_{2}(j)\right)\right)}\right]\;\text{d}x_{nk}.
\end{dmath*}
\end{theorem}
\normalsize  
Based on Theorem \ref{theorem1}, the coverage probability of obtaining file $n$, $C_{n}$, is determined by plugging $C_{nk}$ into (\ref{total.cov}). Furthermore, the coverage probability of obtaining any file $m$ in the content library can be determined by replacing the caching probabilities of file $n$ across different tiers with those of file $m$ in Theorem \ref{theorem1}. Next, for the special case of having $\alpha_{j} = \alpha,\;\forall j$, we provide the coverage probability of obtaining file $n$.
\begin{corollary}\label{corollary1}
The coverage probability of obtaining file $n$ when $\alpha_{j} = \alpha,\; \forall j$, i.e., $\bar{\alpha}_{j} = 1,\; \forall j$, is
\begin{equation}\label{caching_impact}
 C_{n} = \sum_{k=1}^{K}{\dfrac{p_{nk} \lambda_{k}}{\sum_{j=1}^{K}{\lambda_{j} \bar{P}_{j}^{2/\alpha} \left[p_{nj} + a_{j} \left(\rho_{1}(j) + q_{nj} \rho_{2}(j)\right)\right]}}}.
\end{equation}
\end{corollary}
Equation (\ref{caching_impact}) reflects the impact of the probabilistic availability of file $n$ at BSs across different tiers on the coverage probability. Setting $p_{nj} = 1$, $\forall j$, leads back to the achievable coverage of conventional $K$-tier heterogeneous networks \cite{jo2012heterogeneous}, where the requested files were assumed to be available at all BSs. On the other hand, by setting $K = 1$, we obtain the coverage probability of the single tier CHN 
\begin{align}
C_{n} = \dfrac{p_{n1}}{p_{n1} + a_{1} \left(\rho_{1}(1)+ q_{n1} \rho_{2}(1)\right)},
\end{align}
which is consistent with the obtained result for the single tier scenario studied in [\cite{krishnan2015distributed}, Theorem $1$].
\section{Local Delay Analysis}
 Conditioned on the association of the typical user with the $k$-th tier to obtain file $n$, the local delay is defined as the mean number of time slots required for the typical user to successfully receive file $n$ from its serving BS. We assume that if the typical user fails to decode the transmitted file $n$ in a certain time slot, the file is retransmitted in the next time slot. The conditional success probability in an arbitrary time slot is defined as $P_{C_{nk}\left.\right| \phi} = \mathbb{P}\left(\text{SIR}_{k}(x_{nk}) > \tau \left.\right| \phi,x_{nk}\right)$, where $\phi = \cup_{j \in \mathcal{K}}\phi_{j}$. Thus, due to the independence of conditional success events over time, the conditional local delay is a geometric random variable with mean $1/P_{C_{nk}\left.\right| \phi}$ \cite{nie2015hetnets}. Hence, the local delay of obtaining file $n$ can be expressed as
\begin{align}\label{delay.nk}
D_{nk} = \mathbb{E}_{\phi,x_{nk}}\left[\dfrac{1}{\mathbb{P}\left(\text{SIR}_{k}(x_{nk}) > \tau \left.\right| \phi,x_{nk}\right)}\right],
\end{align}
thus from the total probability law, the local delay of obtaining file $n$ is given by
\begin{align}\label{delay.total}
D_{n} = \sum_{k=1}^{K}{A_{k} D_{nk}}.
\end{align}

By applying similar analysis as in (\ref{cov.nk}), $D_{nk}$ can be expressed as  
\small
\begin{align}
\nonumber D_{nk} &= \int_{0}^{\infty} \mathbb{E}_{\phi}\left[\prod_{j=1}^{K}{\dfrac{1}{\mathcal{L}_{I_{j}}\left(\dfrac{x_{nk}^{\alpha_{k}}\tau}{P_{k}}\left.\right| \phi_{j},x_{nk}\right)}}\right] f_{X_{nk}}(x_{nk}) \;\text{d}x_{nk}\\
\label{dnk.mid}&\overset{(a)}{=}\int_{0}^{\infty} \prod_{j=1}^{K}{\mathbb{E}_{\phi_{j}}\left[\dfrac{1}{\mathcal{L}_{I_{j}}\left(\dfrac{x_{nk}^{\alpha_{k}}\tau}{P_{k}}\left.\right| \phi_{j},x_{nk}\right)}\right]} f_{X_{nk}}(x_{nk}) \;\text{d}x_{nk},
\end{align}
\normalsize

where (a) follows from the independence of the PPPs of different tiers. In order to characterize $D_{nk}$, we shall now calculate the expectation inside  the integral in (\ref{dnk.mid}) with respect to each PPP. We first derive the expectation with respect to the PPP of the serving tier, which can be determined by considering the conditional Laplace transform of the interference power from outside and inside the circle of radius $x_{nk}$. Hence, from (\ref{Iout.mid}), we have

\small
\begin{align}\label{exp.k.out}
\nonumber\mathbb{E}_{\phi_{k}}\left[\dfrac{1}{\mathcal{L}_{I_{k\left.\right| \phi_{k},x_{nk}}}^{out}}\right] &= \text{exp}\left[2 \pi \lambda_{k} a_{k} \int_{x_{nk}}^{\infty} \dfrac{\tau y}{\tau - \tau a_{k} + \left(\dfrac{y}{x_{nk}}\right)^{\alpha_{k}}}\text{d}y\right] \\
&\overset{(a)}{=} \text{exp}\left[\pi \lambda_{k} a_{k} \rho_{3}(k) x_{nk}^2\right],
\end{align}
\normalsize
where (a) follows from the change of variables $u = \left(y/x_{nk}\right)^2 \tau^{-2/\alpha_{k}}$ and $\rho_{3}(m)$ is given by
\begin{align}
\rho_{3}(m)= \tau^{2/\alpha_{m}}\int_{\tau^{-2/\alpha_{m}}}^{\infty} \dfrac{1}{1 - a_{m} + u^{\alpha_{m}/2}}\; \text{d}u.
\end{align} 
Similarly, considering the interference from inside the circle of radius $x_{nk}$, we have
\begin{align}\label{exp.k.ins}
\mathbb{E}_{\phi_{k}}\left[\dfrac{1}{\mathcal{L}_{I_{k\left.\right| \phi_{k},x_{nk}}}^{ins}}\right] = \text{exp}\left[\pi \lambda_{k} q_{nk} a_{k} \rho_{4}(k) x_{nk}^2\right],
\end{align}
where $\rho_{4}(m)$ is given by
\small
\begin{align}
\rho_{4}(m)= \tau^{2/\alpha_{m}}\int_{0}^{\tau^{-2/\alpha_{m}}} \dfrac{1}{1 - a_{m} + u^{\alpha_{m}/2}}\; \text{d}u.
\end{align}
\normalsize

Next, we derive the expectation inside the integral in (\ref{dnk.mid}) with respect to the PPP of the $j$-th tier, $\forall j \neq k$. The corresponding expressions to (\ref{exp.k.out}) and (\ref{exp.k.ins}) for the $j$-th tier case can be found is a similar way, and given respectively by

\small
\begin{align}\label{exp.j.out}
\mathbb{E}_{\phi_{j}}\left[\dfrac{1}{\mathcal{L}_{I_{j\left.\right| \phi_{j},x_{nk}}}^{out}}\right] = \text{exp}\left[\pi \lambda_{j} a_{j} \bar{P}_{j}^{2/\alpha_{j}} \rho_{3}(j) x_{nk}^{2/\bar{\alpha}_{j}} \right],\; \forall j \neq k,
\end{align}
\begin{align}\label{exp.j.ins}
\mathbb{E}_{\phi_{j}}\left[\dfrac{1}{\mathcal{L}_{I_{j\left.\right| \phi_{j},x_{nk}}}^{ins}}\right] = \text{exp}\left[\pi \lambda_{j} q_{nj} a_{j} \bar{P}_{j}^{2/\alpha_{j}} \rho_{4}(j) x_{nk}^{2/\bar{\alpha}_{j}} \right],\; \forall j \neq k.
\end{align}
\normalsize
From (\ref{exp.k.out}), (\ref{exp.k.ins}), (\ref{exp.j.out}) and (\ref{exp.j.ins}) along with taking into account that $\bar{P}_{k} = \bar{\alpha}_{k} = 1$, we obtain that
\small
\begin{align}\label{exp.all}
\mathbb{E}_{\phi_{j}}\left[\dfrac{1}{\mathcal{L}_{I_{j\left.\right| \phi_{j},x_{nk}}}}\right] = \text{exp}\left[\pi \lambda_{j} a_{j} \bar{P}_{j}^{2/\alpha_{j}} x_{nk}^{2/\bar{\alpha}_{j}} \left(\rho_{3}(j) + q_{nj} \rho_{4}(j)\right)\right],\forall j.
\end{align}
\normalsize 
Plugging (\ref{exp.all}) into (\ref{dnk.mid}) yields the local delay of obtaining file $n$ when the typical user is associated with the $k$-th tier, as established by the following theorem.
\begin{theorem}\label{theorem2}
The local delay of obtaining file $n$ when the typical user is associated with the $k$-th tier is
\small
\begin{dmath}\label{dnk.final}
D_{nk}=\int_{0}^{\infty} \dfrac{2 \pi p_{nk} \lambda_{k} x_{nk}}{A_{k}} \text{exp}\left[-\pi \sum_{j=1}^{K}{\lambda_{j} \bar{P}_{j}^{2/\alpha_{j}} x_{nk}^{2/ \bar{\alpha}_{j}}\left( p_{nj} - a_{j} \left( \rho_{3}\left(j\right) + q_{nj} \rho_{4}\left(j\right)\right)\right)}\right]\;\text{d}x_{nk}.
\end{dmath}
\end{theorem}
\normalsize
Based on Theorem \ref{theorem2}, the local delay of obtaining file $n$, $D_{n}$, is determined by plugging $D_{nk}$ (\ref{dnk.final}) into (\ref{delay.total}). Now, for the special case of having $\alpha_{j} = \alpha,\;\forall j$, we provide the local delay of obtaining file $n$ in the following Corollary.
\begin{figure*}[htp]
\centerline{
\subfloat[Coverage probability vs. $\tau$.]{\includegraphics[width=2.2in,height= 1.5in]{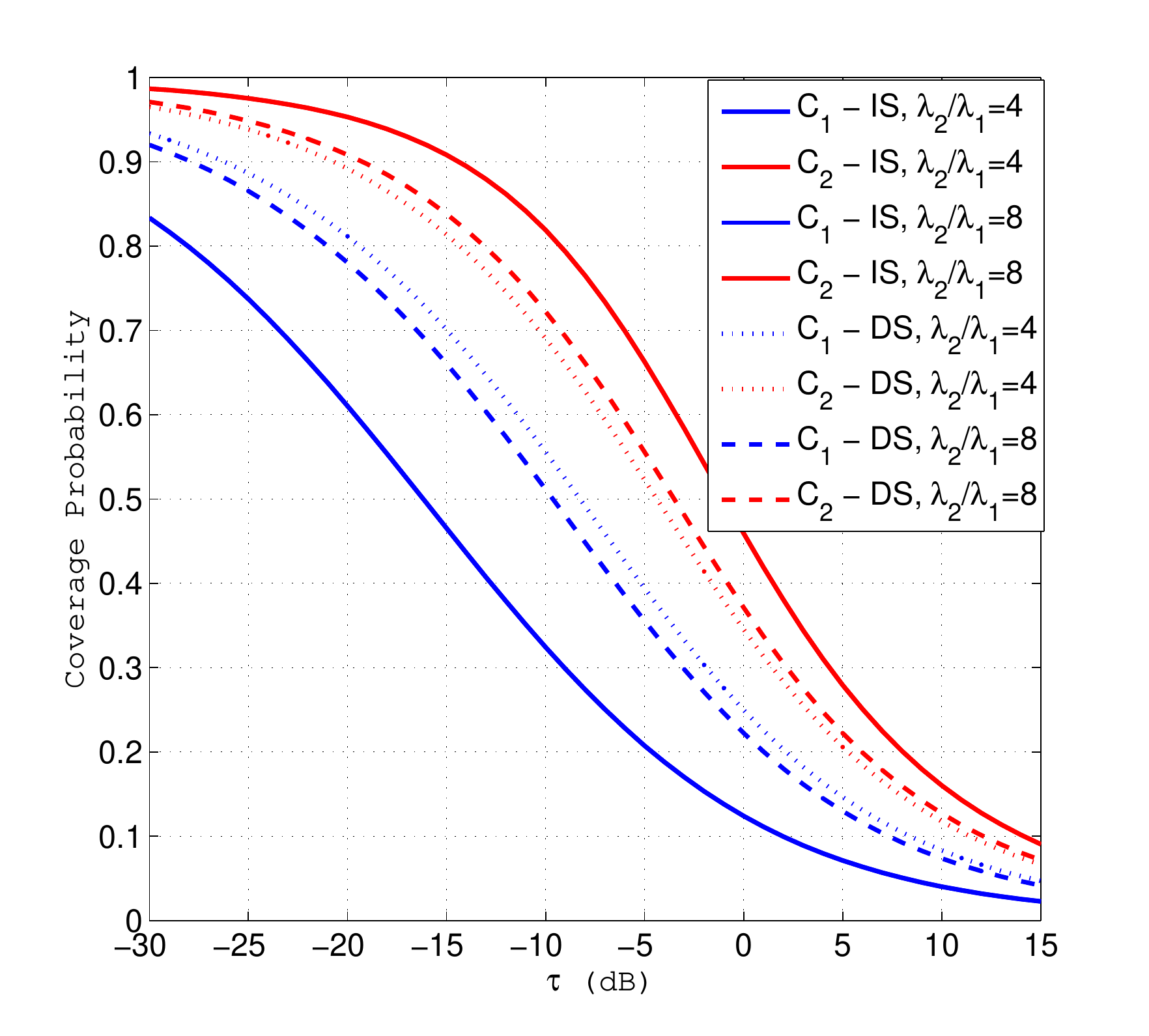}%
\label{1}} \hfil
\subfloat[Local delay vs. $\tau$.]{\includegraphics[width=2.2in,height= 1.5in]{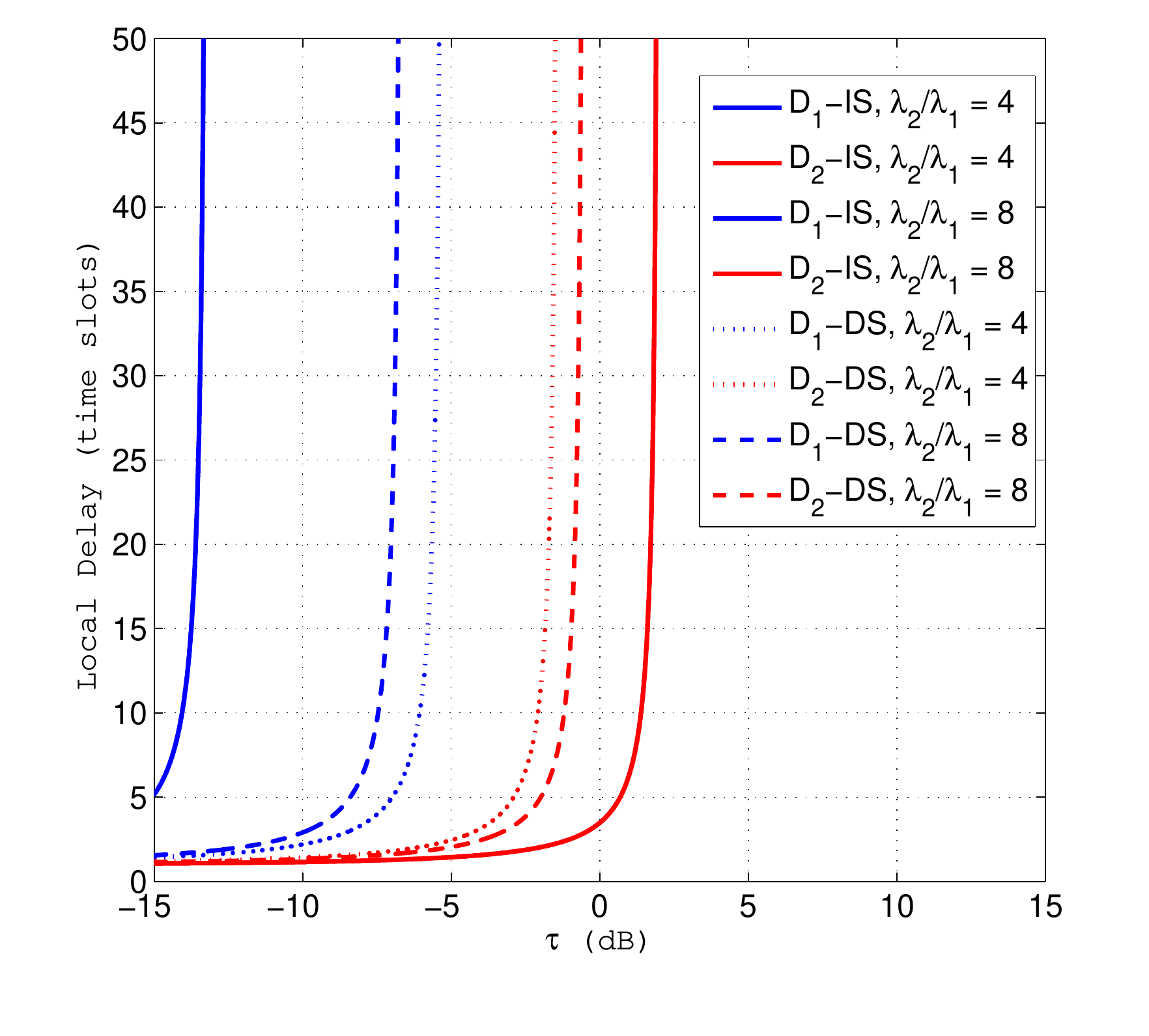}%
\label{2}} \hfil
\subfloat[Coverage probability vs. $a_{1}$ and $a_{2}$.]{\includegraphics[width=1.9in,height= 1.5in]{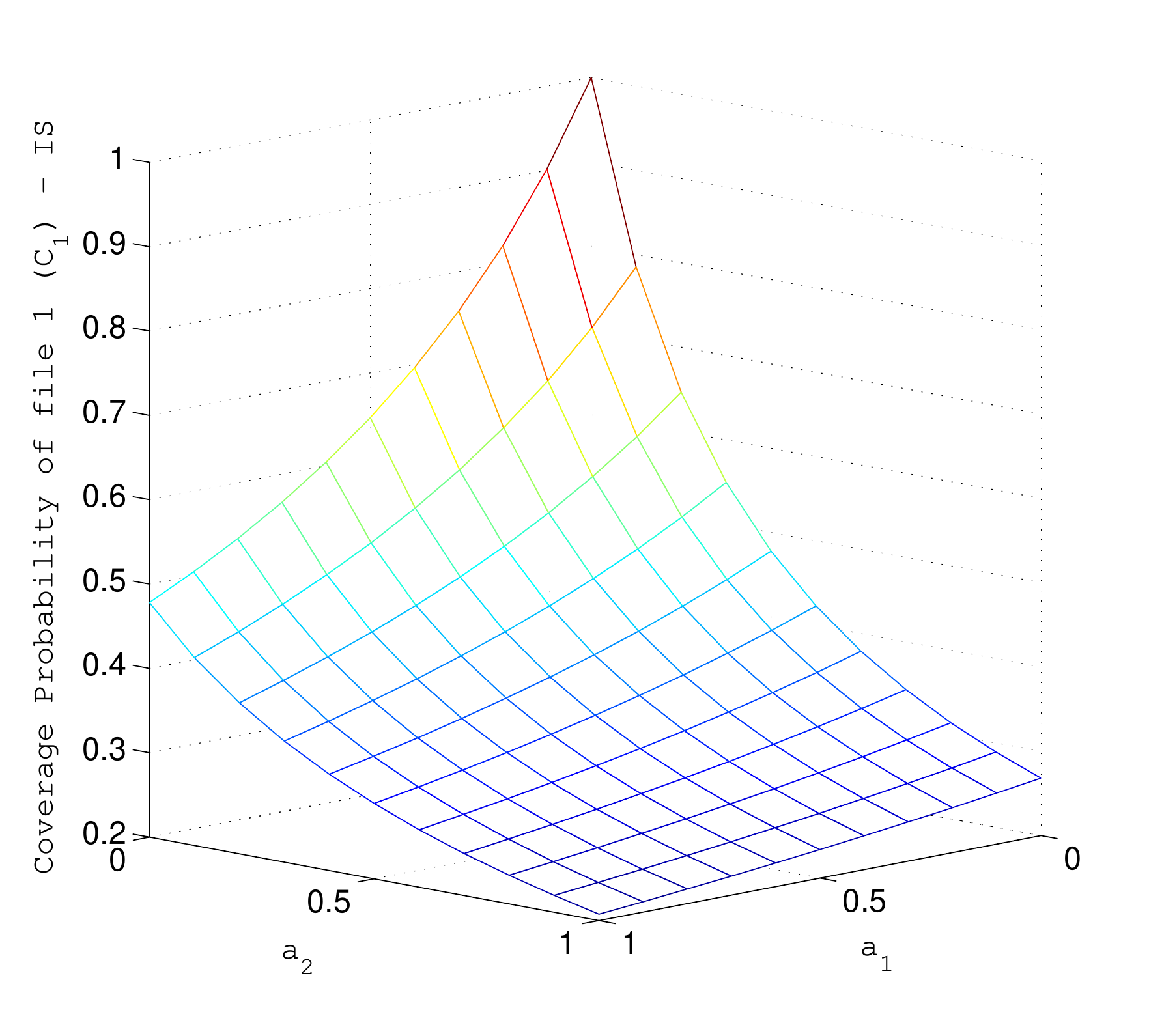}%
\label{3}} } \caption{Numerical Results.}
\end{figure*}
\begin{corollary}\label{corollary2}
The local delay of obtaining file $n$ when $\alpha_{j} = \alpha,\; \forall j$, i.e., $\bar{\alpha}_{j} = 1,\; \forall j$, is
\begin{align}
D_{n} = \sum_{k=1}^{K}{\dfrac{p_{nk} \lambda_{k}}{\sum_{j=1}^{K}{\lambda_{j} \bar{P}_{j}^{2/\alpha} \left[p_{nj} - a_{j} \left(\rho_{3}(j) + q_{nj} \rho_{4}(j)\right)\right]}}}.
\end{align}
\end{corollary}
Setting $K = 1$, leads back to the single tier case where the local delay is given by
\begin{align}
D_{n} = \dfrac{p_{n1}}{p_{n1} - a_{1} \left(\rho_{3}(1) + q_{n1} \rho_{4}(1)\right)},
\end{align}
which is consistent with the obtained result for the single tier scenario studied in [\cite{krishnan2015distributed}, Lemma $5$].
\section{Numerical results}
For clarity of exposition, we focus on a two-tier CHN with a two-file distributed caching scenario, i.e., $K = 2$ and $M = 2$. We further assume that $S_{1} = S_{2} = 1$, i.e., each BS in the network either caches file $1$ or file $2$ only. If not otherwise stated, we use the following parameters, $P_{2} = 0.1$, $P_{1} = 100 P_{2}$, $\lambda_{1} = 1$, $a_{1} = a_{2} = 1$ and $\alpha_{1} = \alpha_{2} = 4$. We consider two different placement caching strategies: 1) Identical placement caching strategy (IS), in which each file has the same caching probability across the BSs of the two tiers, and 2) Different placement caching strategy (DS), in which each file has a different caching probability across the BSs of the two tiers. Particularly, in IS, we use $p_{11} = p_{12} =0.2$ and $p_{21} = p_{22} = 0.8$, whereas in DS, we use $p_{11} = p_{21} = 0.5$, $p_{12} = 0.2$ and $p_{22} = 0.8$. 

In Fig. \ref{1} and \ref{2}, we plot the coverage probability and the local delay, respectively, as a function of $\tau$ for different values of $\lambda_{2}/\lambda_{1}$ and for caching strategies IS and DS. We use $a_{1} = a_{2} = 0.5$ in Fig. \ref{2}. We observe that the coverage probability monotonically decreases as $\tau$ increases, whereas the local delay monotonically increases as $\tau$ increases. Also, file $2$, which has a higher caching probability than file $1$ in both scenarios IS and DS, has a higher coverage and a lower delay in each caching strategy for every $\lambda_{2}/\lambda_{1}$. Third, considering DS scenario, as $\lambda_{2}/\lambda_{1}$ increases, file $2$ experiences a higher coverage and a lower delay while file $1$ experiences a lower coverage and a higher delay. This is due to the fact that increasing $\lambda_{2}/\lambda_{1}$ increases the density of the second tier in which file $2$ has a higher caching probability. Finally, considering IS scenario, as expected, increasing $\lambda_{2}/\lambda_{1}$ has no effect on the achievable coverage and delay of each file. This can be seen from the expressions of coverage and delay in Corollary \ref{corollary1} and \ref{corollary2}, in which coverage and delay do not depend on the spatial densities of tiers for the scenario of IS.

Fig. \ref{3} shows the effect of changing the activity probabilities on the achievable coverage probability of file $1$ in IS case. We use $\lambda_{2} = 4 \lambda_{1}$ and $\tau = -5$ dB. It is observed that the coverage probability is a monotonically decreasing function in both $a_{1}$ and $a_{2}$. This is attributed to the fact that either increasing $a_{1}$ or $a_{2}$ leads to a higher number of active BSs, and hence a higher received interference power at the serving BS which in turn leads to a lower coverage.
\section{Conclusion}
This paper focuses on the performance analysis of a generic $K$-tier cache-aided heterogeneous network. For the considered generic model, we drive closed-form expressions for both the coverage probability and local delay of obtaining each file cached in the network. Our results are consistent with those in the literature for the scenario of single tier cache-enabled heterogeneous network.
As part of our future work, we would like to characterize: 1) the performance of closed access cell association strategy, and 2) the optimal placement caching strategies to maximize the hit probability.

\end{document}